\documentclass[10pt, conference, compsocconf]{IEEEtran}

\usepackage{geometry}
\usepackage{graphicx}
\usepackage{amssymb}
\usepackage{epstopdf}
\usepackage[]{algorithmic}
\usepackage{algorithm}
\usepackage{stmaryrd}
\usepackage{comment}
\usepackage{amsthm}

\newtheorem*{myclaim}{Claim}
\newtheorem*{mythm}{Theorem}

\begin{document}
\title{A Makespan Lower Bound for the Scheduling of the Tiled Cholesky Factorization based on ALAP scheduling}

\author{
\IEEEauthorblockN{Willy Quach}
\IEEEauthorblockA{\'Ecole Normale Sup\'erieure de Lyon\\
Lyon, France\\
Email: willy.quach@ens-lyon.fr}
\and
\IEEEauthorblockN{Julien Langou}
\IEEEauthorblockA{University of Colorado Denver\\
Denver, Colorado, USA\\
Email: julien.langou@ucdenver.edu}
}

\maketitle

\begin{abstract}

Due to the advent of multicore architectures and massive parallelism, the tiled
Cholesky factorization algorithm has recently received plenty of attention and
is often referenced by practitioners as a case study.  It is also
implemented in mainstream dense linear algebra 
libraries.  However, we note that theoretical study of the parallelism of this
algorithm is currently lacking.  In this paper, we present new theoretical
results about the tiled Cholesky factorization in the context of a parallel
homogeneous model without communication costs.  We use standard flop-based
weights for the tasks.  For a $t$-by-$t$ matrix, we know that the critical path of the tiled
Cholesky algorithm is $9t-10$ and that the weight of all tasks is $t^3$. In this context, we prove
that no schedule with less than $0.185 t^2$ processing units can finish in a
time less than the critical path. In perspective, a naive bound gives $0.11 t^2.$
We then give a schedule which needs less than $0.25 t^2+0.16t+3$ processing units to
complete in the time of the critical path.  In perspective, a naive schedule gives $0.50 t^2.$
In addition, given a fixed number of 
processing units, $p$, we give a lower bound on the execution time as follows:
$$\max( \frac{t^{3}}{p}, 
\frac{t^{3}}{p} - 3\frac{t^2}{p} + 6\sqrt{2p} - 7 ,
9t-10).$$
The interest of the latter formula lies in the middle term. Our results stem from the
observation that the tiled Cholesky factorization is much better behaved when
we schedule it with an ALAP (As Late As Possible) heuristic than an ASAP (As
Soon As Possible) heuristic.  We also provide scheduling heuristics which match
closely the lower bound on execution time.  We believe that our theoretical
results will help practical scheduling studies. Indeed, our results enable to better characterize the quality of
a practical schedule with respect to an optimal schedule.

\end{abstract}


\IEEEpeerreviewmaketitle

\section{Introduction}

Most time-consuming tasks performed on supercomputers are linear algebra
operations. With the advent of multicore architectures and massive parallelism,
this results in the necessity to optimize and understand their parallel
executions. Here, we consider the problem of the tiled Cholesky factorization.
The algorithm divides the initial matrix into square sub-matrices, or
\textit{tiles} of the same size. The focus will be on large instances of the
tiled Cholesky factorization, that is where the number of tiles is large, which
allows asymptotical analysis.
To the authors' knowledge, no theoretical non
trivial bound on the execution time of any schedule for the tiled Cholesky
factorization have been found. This motivates this paper.

We note that the tiled Cholesky factorization algorithm has recently received
plenty of attention. Either as an algorithm in
itself~\cite{Gustavson2009,Kurzak2008} or as a case study for task-based
schedulers~\cite{Chan:2008,Agullo2010Comparison,Song2009,agullo2015bridging,journals/concurrency/KurzakLDB10,badia2009}.
Examples of task-based schedulers which have produced papers about the
scheduling of tiled Cholesky factorization are for example
DAGuE~\cite{dague:2012}, KAAPI~\cite{kaapi:2007}, QUARK~\cite{quark:2011}, StarPU~\cite{starpu:2011},
SMPSs~\cite{smpss:2007}, and SuperMatrix~\cite{qqgzc:2009}.
We also note that OpenMP since 3.1 supports task-based parallelism.
The tiled Cholesky
factorization algorithm is also used in practice and is implemented in Dense
Linear Algebra state of the art libraries, for example DPLASMA,
FLAME, and PLASMA.

It is therefore of interest to better understand the parallel execution of the
tiled Cholesky factorization algorithm. In this paper, we neglect
communication costs between processing units. We also ignore any memory
problems a real execution would encounter. Also, we assume homogeneous
processing units.  Also, we assume ideal flop-based weights for the execution
time of the tasks.  While we acknowledge that this is a very unrealistic and
simplistic model, we note that any practical implementations will execute
slower than this model. Therefore, this model provides lower bounds on the
execution time of any parallel execution on any parallel machine.  The lower
bounds that we exhibit are not trivial and are relevant for practical
applications.

We can relate our work to the recent work of Agullo et
al.~\cite{agullo2015bridging} where the authors provide lower bound as well.
The authors consider a more complicated model (heterogeneous) and solve their
problem with an LP formulation. We consider a simpler model (homogeneous)  but we
provide closed-form solutions. Another contribution of our paper is the ALAP
schedule heuristic where tasks are scheduled from the end of the execution as
opposed from the start.

We can also relate our work to the work of Cosnard, Marrakchi, Robert,
and Trystram from 1988 to
1989~\cite{COSNARD1988275,Robert1989159,j38}.  In this work, the
authors have the same model and ask similar questions as ours. A minor
difference is that they are studying the Gaussian elimination while we study
the Cholesky factorization. The main difference is that they study the Level 1
BLAS algorithm which work on columns of the matrix. This algorithm was popular
in the 1980s due to vectorization, nowadays tiled algorithms are much more
relevant.  Also the scheduling of the Level 1 BLAS algorithm seems to be an easier
problem.  In the Level 1 BLAS algorithm, the matrix is partitioned by columns.  The number of created tasks is
$\mathcal{O}(t^2)$ where $t$ is the number of columns of the problem.  In
our case, we partition the matrix by tiles. If we have a $t$-by-$t$ tile
matrix, the number of created tasks is $t^3$. We have tried to apply similar
techniques as in the Level 1 BLAS algorithm study papers to solve the
tiled problem and we have been unsuccessful. We have tried to solve the the Level 1 BLAS algorithm
problem with our
techniques and have obtained the same results.

A few scheduling algorithms exist for the (tiled) Cholesky factorization ; in
practice, the ALAP (As Late As Possible) schedule seems to work well with the
tiled Cholesky factorization. This motivates our study of this
heuristic in Section~\ref{sec:alap}.  In particular, we derive an upper bound
on the number of processing units necessary to reach the critical path of the
algorithm.  Then, we present in Section~\ref{sec:lowerbounds} some lower bounds
on the execution time of any schedule using a given number of processing units
by splitting the task set into two subsets.  In Section~\ref{sec:tightness}, we
analyze the last bound found in Section~\ref{sec:lowerbounds}, and show that it
is nearly optimal by describing a schedule with efficiency close to the
efficiency bound derived from it.

\section{Context, Definitions, Assumptions}
\label{sec:assumptions}

Given a Symmetric Positive Definite (SPD) matrix $A$, the Cholesky
factorization gives a (lower) triangular matrix $L$ such that $A = LL^{T}$. It
is a core operation to compute the inverse of SPD matrices using the Cholesky
inversion. Note that it also allows to solve systems of the form $Ax = b$ by
reducing it to computing solutions of $Ly = b$, and then $L^{T}x = y$.

In order to compute such a factorization using many processing units, we divide the
matrix $A$ into $t \times t$ square tiles of size $n_{b}$. This allows tile
computations, and globally increases the amount of parallelism and the data
locality. Algorithm~\ref{alg1} computes the Cholesky factorization of $A$ using
these blocks. 

\begin{algorithm}
\caption{Tiled Cholesky Factorization}
\label{alg1}
\begin{algorithmic}
\FOR{$k=0$  to  $t-1$}
\STATE $A_{k,k} \leftarrow POTRF(A_{k,k})$ \hspace{45 pt} \COMMENT{$C_{k}$}
\FOR{$i=k+1$ to $t-1$}
\STATE $A_{i,k} \leftarrow TRSM(A_{k,k}, A_{i,k})$\hspace{27 pt}\COMMENT{$T_{i,k}$}
\ENDFOR
\FOR{$j=k+1$ to $n-1$}
\STATE $A_{j,j} \leftarrow SYRK(A_{j,k}, A_{j,j})$\hspace{30 pt}\COMMENT{$S_{j,k}$}
\FOR{$i=j+1$ to $n-1$}
\STATE $A{i,j} \leftarrow GEMM(A_{i,k}, A_{j,k})$\hspace{9 pt}\COMMENT{$G_{i,j,k}$}
\ENDFOR
\ENDFOR
\ENDFOR
\end{algorithmic}
\end{algorithm}

We will rename the tasks corresponding to the POTRF as $C$ or $C_{i}$ with
$1\leq i \leq t$ (as POTRF is a $n_{b}$-by-$n_{b}$ Cholesky factorization), TRSM as $T$ or
$T_{i,j}$ with $1\leq j < i \leq t$, SYRK as $S$ or $S_{i,j}$ with $1\leq j < i
\leq t$, and GEMM as $G$ or $G_{i,j,k}$ with $1\leq k < j < i \leq t$ to refer
to the tasks in Algorithm~\ref{alg1}.

We neglect any communication cost here. Also, tasks $C, T, S, G$ will be
considered as \textit{elementary}: at most one processing unit can execute such
a task at a given time (no divisible load).

The dependencies between the tasks are given by:

\begin{itemize}

\item $C_{j} \rightarrow T_{i,j}, j<i\leq t$;

\item $T_{i,j} \rightarrow S_{i,j}, j<i\leq t$;

\item $T_{i,j} \rightarrow G_{i,k,j}, j<k<i\leq t$;

\item $T_{i,j} \rightarrow G_{k,i,j}, j<i<k\leq t$;

\item $S_{i,j} \rightarrow S_{i,j+1}, j+1<i\leq t$;

\item $S_{i,i-1} \rightarrow C_{i}, 1<i\leq t$;

\item $G_{i,j,j-1} \rightarrow T_{i,j}, 1<j<i\leq t$;

\item $G_{i,j,k} \rightarrow G_{i,j,k+1}, k+1<j<i\leq t$.

\end{itemize}

Figure \ref{Dag55} presents the Directed Acyclic Graph (DAG) of the
dependencies between the tasks of a $5\times 5$ tiled Cholesky Factorization.

\begin{figure}[!t]
\centering
\includegraphics[width=3in]{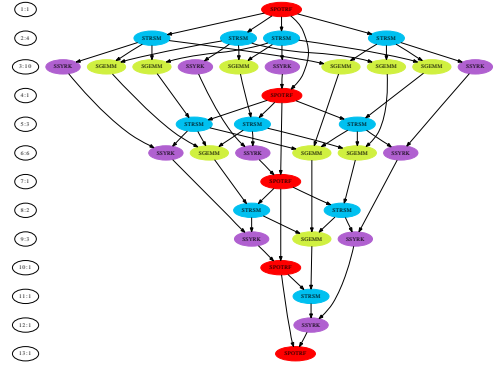}
\caption{DAG of a $5\times 5$ Cholesky factorization}
\label{Dag55}
\end{figure}

For a task $X$, $cp(X)$ will denote the critical path of task $X$, and $w(X)$ its weight.

The number of tasks of each kind is given in Table \ref{numberTasks}.

\begin{table}
\renewcommand{\arraystretch}{1.6}
\caption {Number of tasks}
\label{numberTasks}
\centering
\begin{tabular}{|c|c|}
  \hline
  Type of task & Number of tasks\\
  \hline
  C  & $t$\\
  \hline
  S  & $\frac{t(t-1)}{2}$\\
  \hline
  T  & $\frac{t(t-1)}{2}$ \\
  \hline
  G  & $\frac{t^{3}}{6} - \frac{t^{2}}{2} + \frac{t}{3}$ \\
  \hline
\end{tabular}
\end{table}

Moreover, we will assume that $n_{b}$ is large, so that the weights of the
tasks can be those of Table~\ref{weights}. As a result, for the rest of this
study, we will consider $\frac{1}{3}n_{b}^{3}$ as time unit: executing a POTRF
will take $1$ unit of time, executing a SYRK or a TRSM will take $3$ steps, and
a GEMM $6$ steps.

\begin{table}[!t]
\renewcommand{\arraystretch}{1.6}
\caption {Task weights for the Cholesky factorization}
\label{weights}
\begin{tabular}{|c|c|c|}
  \hline
  Type of Task & number of Flops & Weight (in $\frac{1}{3}n_{b}^{3}*$Flops) \\
  \hline
  POTRF ($C$) & $\frac{1}{3}n_{b}^{3} + \mathcal{O}(n_{b}^{2})$ & 1 \\
  \hline
  TRSM ($T$)& $n_{b}^{3}$ & 3\\
    \hline
  SYRK ($S$)& $n_{b}^{3} + \mathcal{O}(n_{b}^{2})$ & 3 \\
    \hline
  GEMM ($G$)& $2n_{b}^{3} + \mathcal{O}(n_{b}^{2})$ & 6 \\
  \hline
\end{tabular}
\end{table}

We can also count the number of tasks of each kind, and their critical path.
The respective critical paths for the tasks are given in Table \ref{cptask},
which gives the values in Table~\ref{cps}. 

\begin{table}
\renewcommand{\arraystretch}{1.2}
\caption {Critical paths of the tasks}
\label{cptask}
\centering
\begin{tabular}{|c|c|}
  \hline
  Task & Critical Path \\
  \hline
  $C_{i}$  & $C_{i}-T_{i+1,i}-G_{i+2,i+1,i}-T_{i+2,i+1}- . . .$\\
  & $. . . - T_{t,t-1}-S_{t,t-1}-C_{t}$\\
  \hline
  $T_{i,j}$  & $T_{i,j}-G_{i,j+1,j}-T_{i,j+1}- . . . $\\
  &$. . .- T_{i,i-1}-G_{i+1,i,i-1}-T_{i+1,i}- . . .$\\
  \hline
  $S_{i,j}$  & $S_{i,j}-S_{i,j+1}- . . . -S_{i,i-1}-C_{i}- . . .$ \\
  \hline
  $G_{i,j,k}$  & $G_{i,j,k}-G_{i,j,k+1}- . . . -G_{i,j,j-1}-T_{i,j}- . . .$ \\
  \hline
\end{tabular}

\end{table}

\begin{table}
\renewcommand{\arraystretch}{1.2}
\caption {Weight of critical paths of the tasks}
\label{cps}
\centering
\begin{tabular}{|c|c|}
  \hline
  Task & Weight of Critical Path \\
  \hline
  $C_{i}$  & $1$ if $i=t$  \\
  &$9(t-i)-1$ otherwise\\
  \hline
  $T_{i,j}$  & $9(t-j)-2$\\
  \hline
  $S_{i,j}$  & $3(t-j)+1$ if $i=t$ \\
  & $9t-6i-3j-1$ otherwise \\
  \hline
  $G_{i,j,k}$  & $9t-3j-6k-2$ \\
  \hline
\end{tabular}

\end{table}

In particular the critical path for the algorithm ($CP$) is reached for task
$C_{1}$, and equals $9t-10$.

The total work ($TW$) of the Cholesky factorization, that is the sum of the
weights of all tasks of the algorithm, equals $t^{3}$ (with
$\frac{1}{3}n_{b}^{3}$ as the work unit).

Note that as the total work is in $\mathcal{O}(t^{3})$, and as the global
critical path of the algorithm is $\mathcal{O}(t)$, a number $p =
\mathcal{O}(t^{2})$ of processing units is required to achieve a makespan equal to
the critical path. For that reason, our results will be presented with a
number of processing units $p = \alpha t^{2}$, and with $t$ large to allow
such an analysis.

\section{The ALAP (As Late As Possible) schedule} 
\label{sec:alap}

In this section, we analyze the ALAP schedule for the $t \times t$ tiled
Cholesky factorization.  This heuristic seems indeed to perform quite
well.  The schedule is executed as follows: the elementary tasks
(POTRF, SYRK, TRSM, GEMM) are sorted according to their critical path. Then the
tasks with least critical path are set to be executed last; those with least
critical path among the remaining tasks set to be executed the latest possible
so that the previous ones can be executed, etc. Thus, if this schedule has a
makespan $\tau$ and has enough processing units, task $X$ will begin its
execution at time $\tau - cp(X)$ where $cp(X)$ denotes the critical path of
task $X$. Therefore, we study the distribution of the critical paths of the
tasks to understand how many processing units are required to run an optimal
ALAP schedule. 

The section flows as follows. First we analyze the ALAP schedule to obtain the
number of tasks executed at each ticks of the algorithm. The results are
presented in Table~\ref{eq:alap heights}. Then we established simpler lower
bounds and upper bounds to study this function. These bounds are given in
Table~\ref{lowerHeightALAP} and Table~\ref{upperHeightALAP}.  In
Figure~\ref{ALAPbounds}, we plot the exact function and our associated lower
and upper bounds for $t=60$.  We note that our lower and upper bounds are
assymptotically close to the function.  Finally we conclude the section with an
upper bound on the number of processors needed to obtain a makepsan equal to
the critical path.

Figure \ref{ALAP88} shows the execution of an ALAP schedule with sufficiently
many processing units on a $8 \times 8$ tiled Cholesky factorization, with
every rectangle representing a task, and the time on the x-axis.

\begin{figure}[!t]
\centering
\includegraphics[width=3in]{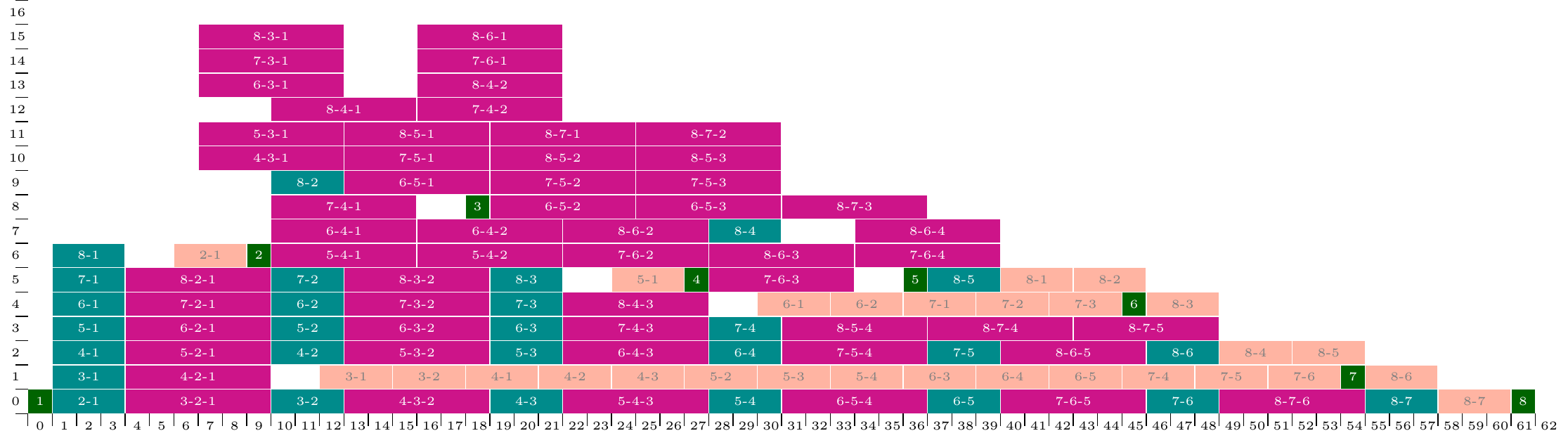}
\caption{ALAP schedule on $8 \times 8$ tiles. Dark green POTRF, light green for
TRSM, salmon for SYRK, and magenta for GEMM.  Time is in the x-axis. Execution
time is the critical path length: $9t-10 = 62$. Number of required processing
units for the schedule is 16.
}
\label{ALAP88}
\end{figure}

More precisely, we want to determine, at time $\tau - K$, where $K$ denotes a
critical path parameter, how many tasks are executed with an ALAP schedule with
sufficiently many processing units. Note that the execution of a task $X$ starts at
$\tau - cp(X)$, but ends at $\tau - cp(X) + w(X)$. Therefore, such a task $X$
should count as being executed for times $\tau - cp(X)$ to $\tau - cp(X) + w(X)
-1$ (we suppose that at time $\tau - cp(X) + w(X)$, the execution of task $X$ has finished).

With that in mind, we can count the number of tasks of each kind (POTRF, SYRK,
TRSM, GEMM) being executed at a time $\tau - K$. 

Counting the tasks C and T is straightforward with the formulas of Table~\ref{cps}. 

To count the tasks S, note that the \textit{S-floor of level $i$}, $i < t$
defined by being the set $\left\{S_{i,j}\right\}_{j<i}$ begins at $K^{(i)}_{min} = 8 +
9(t-i-1) = 9(t-i)$ and ends at $K^{(i)}_{max} = 8+9(t-i-1)+1+3(i-1)-1 = 9t - 6i -4$, and that for any
$K$ in-between, exactly $1$ task of the S-floor of level $i$ is executed. Note that the
S-floor of level $i$ starts (resp. ends) before level $i-1$ starts (resp. ends)
with respect to $K$. Therefore, for any $K$, there are some $i_{min} > 1$ and
$i_{max} < t$, such that at time $\tau - K$, exactly one task from every floor $i$,
$i_{min}\leq i \leq i_{max}$ is executed; we then add the S-floor of level $t$, which starts at $K=2$ and ends at $K = 3t -2$. And these are the only SYRKs being
executed at that time. The expressions of $i_{min}$ and $i_{max}$ follow by noticing that $i_{max}$ (resp. $i_{min}$) corresponds to the largest $i$ such that $cp(S_{i,1}) \geq K$,  \textit{i.e.} the S-floor of level $i_{max}$ starts before $\tau - K$ (resp. the least $i$ such that $cp(S_{i,i-1})\leq K$: the S-floor of level $i_{min}$ finishes before $\tau - K$), and using Table \ref{cps}.

For the tasks G, let us fix $K$. Let $J = \left\{ j | \textmd{ some } G_{i,j,k}
\textmd{ is executed at time } \tau - K \right\}$. Note that $G_{i,j,k}$ is executed at time $\tau - K$ if and only if $K \leq cp(G_{i,j,k}) \leq K+5$, as tasks G require 6 steps of time to be executed. The expression of the critical path of the tasks G gives that $J$ is an integer interval $\llbracket
j_{min}, j_{max} \rrbracket$, and that for all $j \in J$, there is a unique
$k_{j}$ such that $G_{i,j,k_{j}} \textmd{ is executed at time } \tau - K $.
Also, for all $j \in J$, any $G_{i,j,k_{j}}$ where $j < i \leq t$ is executed
at time $\tau - K$ (a whole \textit{G-column} is executed). To compute $j_{min}$ (resp. $j_{max}$), one necessary and sufficient condition is that $cp(G_{i,j_{min},k}) \geq K$ for some $k$ (resp. $cp(G_{i,j_{max},k}) \leq K$ for some $k$), which is equivalent to the fact that $cp(G_{i,j,j-1}) \geq K$  (resp.$cp(G_{i,j,1}) \leq K$).

This reasoning gives Table~\ref{eq:alap heights} where we can find the
following formulas, with $M_{X,K}$ being the number of tasks of type X being
executed at time $\tau - K$ in an ALAP schedule with sufficiently many
processing units.

\begin{table}
\renewcommand{\arraystretch}{1.6}
\caption {The heights of the ALAP schedule}
\label{eq:alap heights}
\centering
\begin{tabular}{|c|}
\hline
\begin{minipage}{7cm}
$M_{X,K}$ is the number of
tasks of type X being executed at time $\tau - K$ in an ALAP schedule with
sufficiently many processing units.
\begin{eqnarray}
\nonumber
M_{C,K} & = &
\left\{\begin{array}{cl}
1&\textmd{ if } K = 9\ell +8,\\
0&\textmd{ else.}
\end{array}\right.\\
\nonumber
M_{T,K} & = &
\left\{\begin{array}{cl}
\ell+1&\textmd{ if } 9\ell + 5 \leq K \leq  9\ell + 7\\
0&\textmd{ else.}
\end{array}\right.\\
\nonumber
M_{S,K} & = &
\left\{\begin{array}{cl}

i_{max} 
- i_{min} +2&\textmd{if } 2\leq K \leq 3t -2\\
\\
i_{max}  
- i_{min} +1 &\textmd{otherwise}
\end{array}\right.\\
\nonumber
M_{G,K} & = &
\sum_{j = j_{min}+1}^{j_{max}}
\left( t - j \right)
\end{eqnarray}
where $i_{max} = \min\left(t, \lfloor \frac{3t}{2} - \frac{K}{6} - \frac{7}{12}
\rfloor \right)$ and $i_{min} = \lceil t - \frac{K}{9}\rceil$ denote the two
extremal S-floors executed at time $\tau-K$, and $j_{max} = \min\left(t-1,
\lfloor 3t - \frac{K}{3} - \frac{8}{3}\rfloor \right)$, $j_{min} = \lceil t -
\frac{K}{9} - \frac{2}{9}\rceil$ are as defined above.
\end{minipage}\\
\hline
\end{tabular}
\end{table}

We can now divide the execution time of the ALAP schedule into three zones. 

In a first zone, both $i_{max}$ and $j_{max}$ are constrained, as they cannot
be greater than $t$ and $t-1$ respectively. One possible interpretation is that
for a bigger instance of Cholesky factorization (with $t' > t$), other S-floors
would have been executed in this zone; in addition to the floors, other
G-columns would have been executed. This zone is delimited by $K < 3t +2 :=
K_{S}$. We will call this zone \textit{Zone 1}.

In a second zone, $i_{max}$ is constrained by $t$, but $j_{max}$ is not
constrained. It corresponds to $K_{S} < K \leq K_{G} := 6t-5$. We will call
this zone \textit{Zone 2}.

In a third zone, $i_{max}$ and $j_{max}$ are not constrained; it is the case as
long as $K > K_{G}$. We will call this zone \textit{Zone 3}.

Summing these formulas give the height of the ALAP schedule we look for. But
because of the ceils and floors, we will not get any clear formula at the end.
As a result, we focus on getting lower bounds and upper bounds on the height,
by using $x\leq \lceil x \rceil < x+1$, and $x-1<\lfloor x \rfloor \leq x$. Let
us note $h(t,K)$ the number of tasks being executed at time $\tau-K$ by an ALAP
schedule with sufficiently many processing units. We have $h(t,K) = M_{C,K} +
M_{T,K} + M_{S,K} + M_{G,K}$. That gives the formulas in
Tables~\ref{lowerHeightALAP} and~\ref{upperHeightALAP}.

\begin{table}
\renewcommand{\arraystretch}{1.6}
\caption {Lower bound on the height of the ALAP schedule}
\label{lowerHeightALAP}
\centering
\begin{tabular}{|c|c|}
  \hline
  Zone & Lower bound on height \\
  \hline
  1  & $\frac{K^{2}}{162}-\frac{5K}{162}-\frac{25}{81}$\\
  \hline
  2  & $\frac{K^{2}}{162}-\frac{16K}{162}+\frac{t}{2}-\frac{289}{324}$\\
  \hline
  3  & $-\frac{4K^{2}}{81}+\frac{2Kt}{3}-\frac{197K}{162}-2t^{2}+\frac{119t}{18}-\frac{587}{108}$ \\
  \hline
\end{tabular}
\end{table}

\begin{table}
\renewcommand{\arraystretch}{1.6}
\caption {Upper bound on the height of the ALAP schedule}
\label{upperHeightALAP}
\centering
\begin{tabular}{|c|c|}
  \hline
  Zone & Upper bound on height \\
  \hline
  1  & $\frac{K^{2}}{162}+\frac{31K}{162}+\frac{155}{81}$\\
  \hline
  2  & $\frac{K^{2}}{162}+\frac{2K}{81}+\frac{t}{2}+\frac{755}{324} $\\
  \hline
  3  & $-\frac{4K^{2}}{81}+\frac{2Kt}{3}-\frac{107K}{162}-2t^{2}+\frac{83t}{18}+\frac{37}{108}$ \\
  \hline
\end{tabular}
\end{table}

In Figure~\ref{ALAPbounds}, we plot the exact distribution of the execution of
the tasks, and the upper and lower bound functions.  We note that our lower and
upper bounds are assymptotically close to the function. From this figure, we
see that for $t=60$ tiles, we have a schedule that executes in $9t-10$ (=530)
for 907 processing units.  Our upper bound (which is simpler to analyze)
guarantees that we need at most 913 processing units.

\begin{figure}[!t]
\centering
\includegraphics[width=3in]{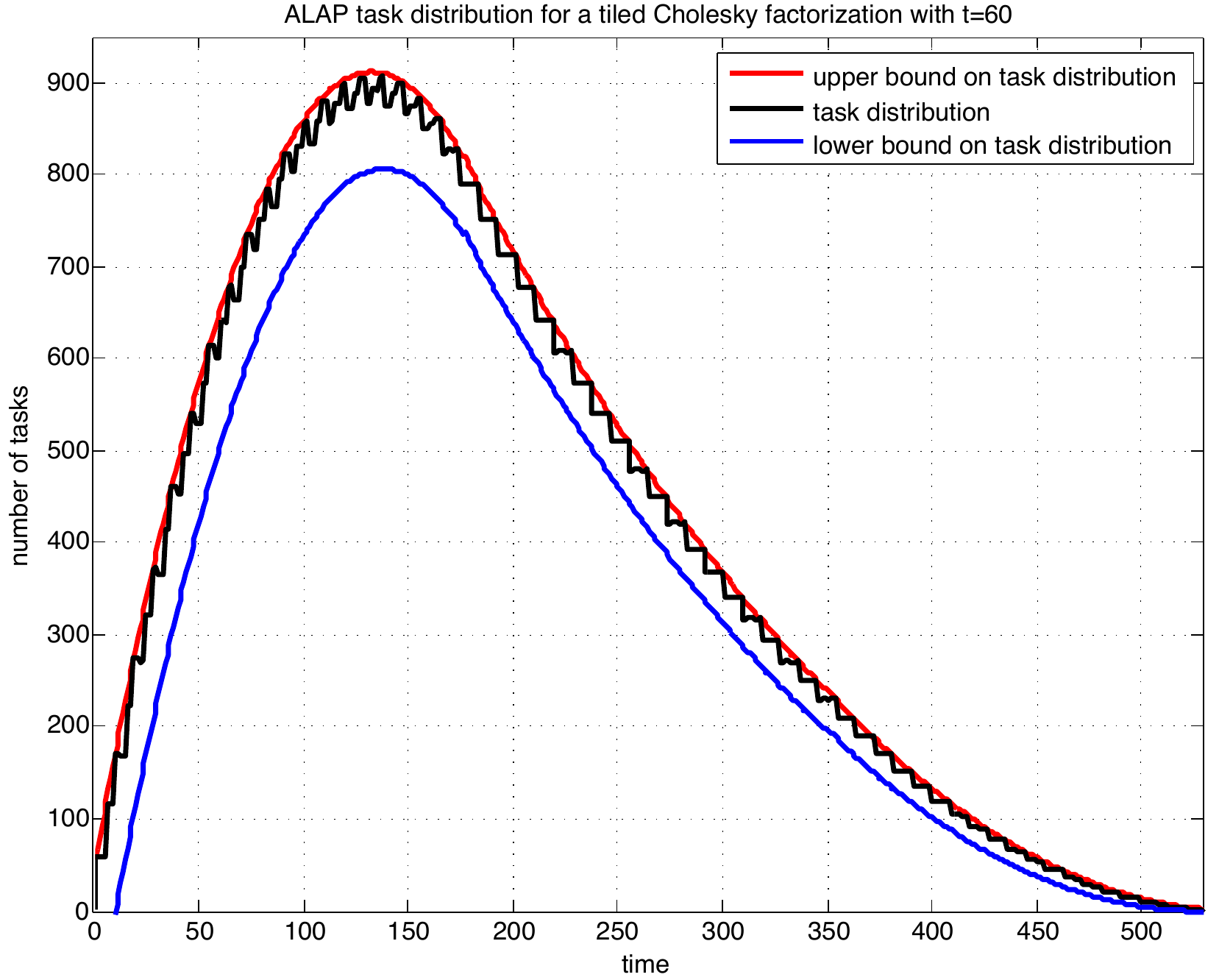}
\caption{ALAP distribution for $t=60$ tiles. In black is the exact distribution. Red is the upper bound. Blue is the lower bound.}
\label{ALAPbounds}
\end{figure}

Note that the maximum height of the distribution (that is the minimum number of
processing units required to run the ALAP schedule optimally) is reached in
\textit{Zone 3}, \textit{i.e.} where $K > 6t - 5$. That phenomenon is verified
experimentally. It can also be deduced from the fact that the height in the two
other zones is essentially a degree 2 polynomial in $K$ with positive leading
coefficient (which comes from the number of G executed); therefore by convexity
the maximum height there is reached at the extremities of the zones. The result
follows by observing that the height in the third zone is increasing at its
beginning.

From that, we obtain the maximum height of the distribution, which gives the
minimal number of processing units required to run the ALAP schedule:

$$\forall K, h(t,K) \leq \frac{1}{4}t^2+\frac{11}{72}t+\frac{13225}{5284}.$$

We deduce an upper bound on the number of processing units required to reach
the critical path of the Cholesky factorization:

\begin{mythm}
\label{th:1}
The ALAP schedule of a $t \times t$ tiled Cholesky factorization completes in
the critical path time ($9t-10$) using less than 
$p = 0.25 t^2+0.16t+3$ processing units.
\end{mythm}

We note that this results is much better than what an ASAP schedule would give.
The ASAP schedule of a $t \times t$ tiled Cholesky factorization completes in
the critical path time ($9t-10$) critical path using  $p = 0.5 t(t-1)$ processing
units. The analysis is easy and is based on the fact that, with an ASAP
schedule, after the first POTRF and the first $t-1$ TRSM, one has $0.5 (t-2)(t-1)$
SYRK and GEMM tasks to execute.

\section{Lower bounds}
\label{sec:lowerbounds}

We now consider lower bounds on the makespan of the Cholesky factorization. To
that purpose, we split the task set into two parts, and combine lower bounds on
the execution time of both parts.

It is well known and clear that the makespan of any schedule working on a given algorithm
with $p$ processing units is greater than $\max(TW/p, CP)$ (where $TW$ is the
total work necessary to finish the algorithm, $CP$ its critical path). Indeed,
$TW/p$ corresponds to a schedule that achieves full parallelism during its
whole execution with p processing units; and no schedule can execute the
algorithm faster than $CP$.

The idea here will be to combine those two bounds: let $K$ be a critical path parameter.
Let $E_{2}$ be the set of tasks $X$ such that $cp(X) - w(X) \leq K$, and $E_{1}$ be its complement.

Note that in the general case, $E_{2}$ contains some tasks with critical path
$\geq K$. Indeed, any task of $E_{2}$ without parents in $E_{2}$ necessarily
has critical path $\geq K$ (else any of its parents would have been in $E_{2}$
too). Also, as long as $E_{1}$ and $E_{2}$ are not empty, some task in $E_{1}$
has a child in $E_{2}$ (the two sets cannot be disconnected).  

Note also that all the tasks in $E_{1}$ have critical path $>K$.

Define for a set $E$ of tasks, the total weight of $E$ $w(E) = \sum_{X \in E}
w(X)$ the sum of the weights of the tasks in $E$. The basic idea is that one
will essentially require time $\geq w(E_{1})/p$ to execute $E_{1}$ and time
$\geq K$ to execute $E_{2}$.

This results in the following claim to prove our lower bound:

\begin{myclaim}
Let $K$ be a critical path parameter. Let $E_{2}$ be the set of tasks $X$ such that $cp(X) - w(X) \leq K$, and $E_{1}$ be its complement.
Then the makespan of any schedule for using $p$ processing units is greater than
$$
\max(CP, \frac{w(E_{1})}{p} + K)
$$
\end{myclaim}

\begin{proof}
Suppose by contradiction that there exists a schedule executing the algorithm
in time $\tau < \frac{TW-w(E_{1})}{p} + K$. If $\tau < K$, then the schedule
cannot execute $E_{2}$ (defined above) in time $\tau$, as $E_{2}$ contains some
tasks with critical path $\geq K$, which gives a contradiction. Therefore $\tau
\geq K$. 

Then, let us write $\tau = (\tau -K)+K$. Then $\tau-K < \frac{w(E_{1})}{p}$ by
hypothesis. So after time $\tau -K$, $E_{1}$ cannot be fully executed,
according to the naive bound. Therefore, there is a task from $E_{1}$ that is
not fully completed at time $\tau-K$, and the remaining tasks cannot be
executed in time $K-1$ (as if $X \in E_{1}$, $cp(X)-w(X)>K$ by definition of
$E_{1}$, therefore $X$ has a son with critical path $\geq K$), contradiction.
\end{proof}

A remark here : why not simply take $\widetilde{E_{1}}$ as the set of tasks
that have critical path $> K$, and $\widetilde{E_{2}}$ the set of tasks that
have critical path $\leq K$? Because then the argument would be erroneous, as
tasks $X$ with $cp(X) > K$ and $cp(X)-w(X) < K$ need not to be fully completed
when starting the execution of $\widetilde{E_{2}}$ (for instance, their
execution could be half-done). As a result, we remove these tasks from
$\widetilde{E_{1}}$ and put them in $\widetilde{E_{2}}$, obtaining the previous
definition.

To use this claim, we need an expression of the work $w(E_{1})$ for a some
parameter $K$; we reduce this problem to computing $w(E_{2})$, as $w(E_{1}) +
w(E_{2}) = TW$. Note that the problem is very similar to the calculation of the
height of the ALAP schedule in Section \ref{sec:alap}, but this time, we want
all the tasks with critical path lower than the ones counting for the height of
the ALAP. As a result, we name this quantity the \textit{cumulative
distribution} of the tasks.

In order to simplify the calculations and the results, we focus on the GEMMs
tasks only, which forms the prominent set of tasks of the Cholesky
factorization, both in terms of number (there are $t^{3}/6 + \mathcal{O}(t^2)$
GEMMs), and in terms of work (they gather work $t^{3} + \mathcal{O}(t^2)$)
which asymptotically respectively are the global number of tasks, and the total
work of the algorithm.

Thus, we count, for a given critical path parameter $K$, the number $D_{G}$ of GEMM tasks $X$ such that:
\begin{equation}
\label{eqLowerBound}
cp(X) - w(X) \leq K.
\end{equation}
 
Recall that we assumed that the GEMMs have a weight of 6 (Table~\ref{weights}) and that we proved that the critical path of the task $G_{i,j,k}$ is
$9t-3j-6k-2$ (Table~\ref{cps}). Note then that similarly to counting the
distribution in Section~\ref{sec:alap}, the integer set $J := \left\{ j|
\exists i,k \textmd{ such that } G_{i,j,k} \textmd{ satisfies
(\ref{eqLowerBound})}\right\} $ is convex, hence some $j_{min}$ and $j_{max}$
such that $J = \llbracket j_{min}, j_{max} \rrbracket$. Also, note that for any
$j \in J$, the set $\left\{ k| G_{i,j,k} \textmd{ satisfies
(\ref{eqLowerBound})}\right\}$ is also convex, hence some $k^{(j)}_{min}$ and
$k^{(j)}_{max}$. Note that the formula $cp(G_{i,j,k}) = 9t-3j-6k-2$ immediately
gives $j_{max} = t-1$ and $k^{(j)}_{max} = j-1 \textmd{ } \forall j$. Remark
that $j_{min}$ is determined by the exact same equation as the one considered
in Section \ref{sec:alap}, so that $j_{min} = \lceil t - \frac{K}{9} -
\frac{2}{9}\rceil$. And an easy calculation leads to $k^{(j)}_{min} = max(1,
\lceil 3t/2 - j/2 - K/6 - 7/6 \rceil)$, as $k^{(j)}_{min}$ has to be positive.
So, for fixed j, there are $j-k^{(j)}_{min}$ distinct couples $(j,k)$ that
satisfy Equation~(\ref{eqLowerBound}). And for each admissible couple $(j,k)$,
every $G_{i,j,k}, j<i\leq t$ satisfy Equation~(\ref{eqLowerBound}).

We obtain that the cumulative distribution of the GEMMs is given by:

$$
D_{G,K} = \sum_{j = j_{min}+1}^{j_{max}} (t-j) (j-k^{(j)}_{min}).
$$

The previous claim then gives lower bounds on the execution time for a fixed
number $p$ of processing units. But as we have only considered the GEMMs tasks, we
have to slightly modify our statement.

\begin{myclaim}
Under the same assumptions as in the previous claim, the makespan of any schedule for using $p$ processing units is greater than
$$
\max(CP, \frac{w_{G}(E_{1})}{p} + K),
$$
where $w_{G}(E_{1})$ denotes the total work of the GEMMs in $E_{1}$.
\end{myclaim}

Note that this claim gives a worse bound than the previous one, as
$w_{G}(E_{1})< w(E_{1})$. Also, the previous proof still stands: it suffices to
replace $w(E_{1})$ by $ w_{G}(E_{1})$ (while keeping the same $E_{2}$).

To use this new result, recall that the total work from the GEMMs is $TW_{G} =
\# GEMMs \times w(G) = t^{3} - 3t^{2} + 2t$. Then, for a fixed parameter $K$,
we have $w_{G}(E_{1}) = TW_{G} - 6 D_{G,K}$.

With that in hand, every parameter $K$ gives a lower bound on the makespan of
any schedule using $p$ processing units.

For instance, we find that as long as $k^{j}_{min} = \lceil 3t/2 - j/2 - K/6 -
7/6 \rceil$, that is as long as $K < 6t - 4$, and with $p = \alpha t^{2}$
processing units, the number of tasks in $E_{2}$ is:
\begin{equation}
\label{sizeE2}
\frac{(K-7)(K^2+10K+16)}{2916}.
\end{equation}

Then, the lower bound associated to $K$ is:
$$
\frac{t^{3} - 3t^{2} + 2t - 6\frac{(K-7)(K^2+10K+16)}{2916}}{p} + K.
$$

At this point, we want to pick the best lower bound possible among all the
parameters $K$. Let us name it $K_{max}$ for instance. 

Experimentally, $K_{max} < 6t - 4$, so that the maximum is reached where
$k^{(j)}_{min} \geq 1$. Also, asymptotically (that is when $t\rightarrow
\infty$), the maximum is reached for $K_{max, \infty} = 9\sqrt{2p}$, which is $< 6t - 4$ as long as $\alpha < 2/9$. We can assume this condition, as we will see that our bound is only relevant for $\alpha \leq 0.186$ : otherwise the naive bound is better than ours. This gives us a correct lower bound, even if in practice the maximum is
not reached at the exact same point as in the asymptotical case; but due to the
complexity of the formula with the exact maximum, we first simplify as: Any
schedule working with $p = \alpha t^{2}$ processing units on the $t\times t$
tiled Cholesky factorization has an execution time greater than:
$$
\frac{t^{3}}{p} +6\sqrt{2p} -\frac{3}{\alpha} -7 + \frac{7\sqrt{2}}{3\sqrt{p}}+\frac{2}{\alpha t} + \frac{83}{243\alpha t^{2}}.
$$
And, with some more simplifications, we get the following theorem.

\begin{mythm}[Lower bound on the makespan of the Cholesky factorization]
Any schedule working with $p$ processing units on the $t\times t$ tiled Cholesky factorization has an execution time greater than:
$$\max( \frac{t^{3}}{p}, 
\frac{t^{3}}{p} - 3\frac{t^2}{p} + 6\sqrt{2p} - 7 ,
9t-10).$$
\end{mythm}

As $K_{max} < 6t-4$, the ALAP height in $E_{2}$ is a non-decreasing function of
$K$. Thus, the maximum ALAP height of $E_{2}$ is located at $K = K_{max}$.  It
turns out that the maximum ALAP height in $E_{2}$ (obtained with Tables
\ref{lowerHeightALAP}, \ref{upperHeightALAP} in Section \ref{sec:alap}) is
asymptotically equal to the number of processing units available. In
particular, we can assume $\alpha > \epsilon$ for some fixed parameter
$\epsilon > 0$, for this bound to be a non-negligible improvement over the
naive one (that is we assume the size of $E_{2}$ is not negligible).

Then, this result improves the naive bound with a term $6\sqrt{2p} +
\mathcal{O}(1)$, that does not depend on the size $t$ of the problem, which is
quite a surprising result at first glance.

Let us propose an explanation for this phenomenon. As $K_{max} < 6t-4$, the
number of tasks in $E_{2}$ given in Equation~\ref{sizeE2} only depends on $K_{max}$, not
on $t$. As the gain from the naive bound comes from the fact that $E_{2}$
cannot be well parallelized (as the limiting factor is considered to be the
critical path there), this results in a gain which asymptotically does not
depend on $t$ (the negligible terms are due to the fact that we only considered
GEMMs here). 

Also, with the naive bound, we knew that at least $p = 0.11t^{2}$ processing
units were necessary to reach the critical path. With our new bound, we know
that we need at least $p = 0.185t^{2}$ processing units to reach it.
In other words, the naive bound is better when $\alpha > 0.186$, which justifies the previous assumption.

\section{Analysis of the tightness of the lower bound}
\label{sec:tightness}

The bound exhibited in the previous section improves the previously known ones. 
This raises a question: can we still improve it? That is, is the bound tight?
Here we analyze different schedules and compare their performance with the upper bound on
performance that we derive from our lower bound on time.

In this section, we study some existing schedules. Since our
execution model is theoretical (assuming that POTRFs take 1 unit of time,
TRSMs and SYRKs 3 units, and GEMMs 6) we simulate these schedules with
the same assumptions in order to build a consistent framework for
comparison. Therefore this is a theoretical study.

We focus on three different schedules:

\begin{itemize}
\item the right-looking Cholesky algorithm with multithreaded BLAS, this schedule is implemented in LAPACK for example;
\item a schedule from Kurzak et al. described in~\cite{journals/concurrency/KurzakLDB10};
\item the ALAP schedule mentioned earlier with a list scheduling.
\end{itemize}

The LAPACK schedule is the right-looking Cholesky algorithm with multithreaded BLAS. 
This boils down to synchronizing all processing units
at the
end of every loop in the Cholesky factorization (Algorithm \ref{alg1}, Section
\ref{sec:assumptions}). More precisely, one processing unit executes the POTRF,
while all the other ones wait. When it has finished, they all execute TRSMs if
some are available, and wait otherwise. Then, they execute the SYRKs and the
GEMMs the same way. As a result, the LAPACK schedule suffers from huge
synchronization needs, and therefore should result in quite poor performance
overall. This is bulk synchronous parallelism or the fork-join approach.

The schedule described in~\cite{journals/concurrency/KurzakLDB10} is a variant
of the left-looking Cholesky factorization: it assigns rows to the processing
units, which execute their assigned tasks as soon as possible.

We also consider a schedule based on the ALAP heuristic. Therefore we schedule
the task from the end to the start using a list schedule and priority policy
based on the (ALAP) critical path of a task.

A comparison of the speedups is plotted in Figure~\ref{fig1} for a Cholesky
factorization with 40 tiles. The LAPACK schedule shows quite poor performance
as expected, and the two others schedule performs much better. However, the gap
between their performance curve and the upper bound is quite close. 

The horizontal black curve represent the critical path bound. No schedule can
execute faster than the critical path. For $t=40$, the critical path is $350$.
We see that the green curve reaches the critical path at $p=275$ processing
units.  This means that any schedule which completes in the critical path time
has to have at least $p=275$ processing units.  We see that the red curve
reaches the critical path at $p=343$ processing units.  This means that the
ALAP schedule completes in the critical path time with $p=343$ processing
units.

\begin{figure}[!t]
\centering
\includegraphics[width=3in]{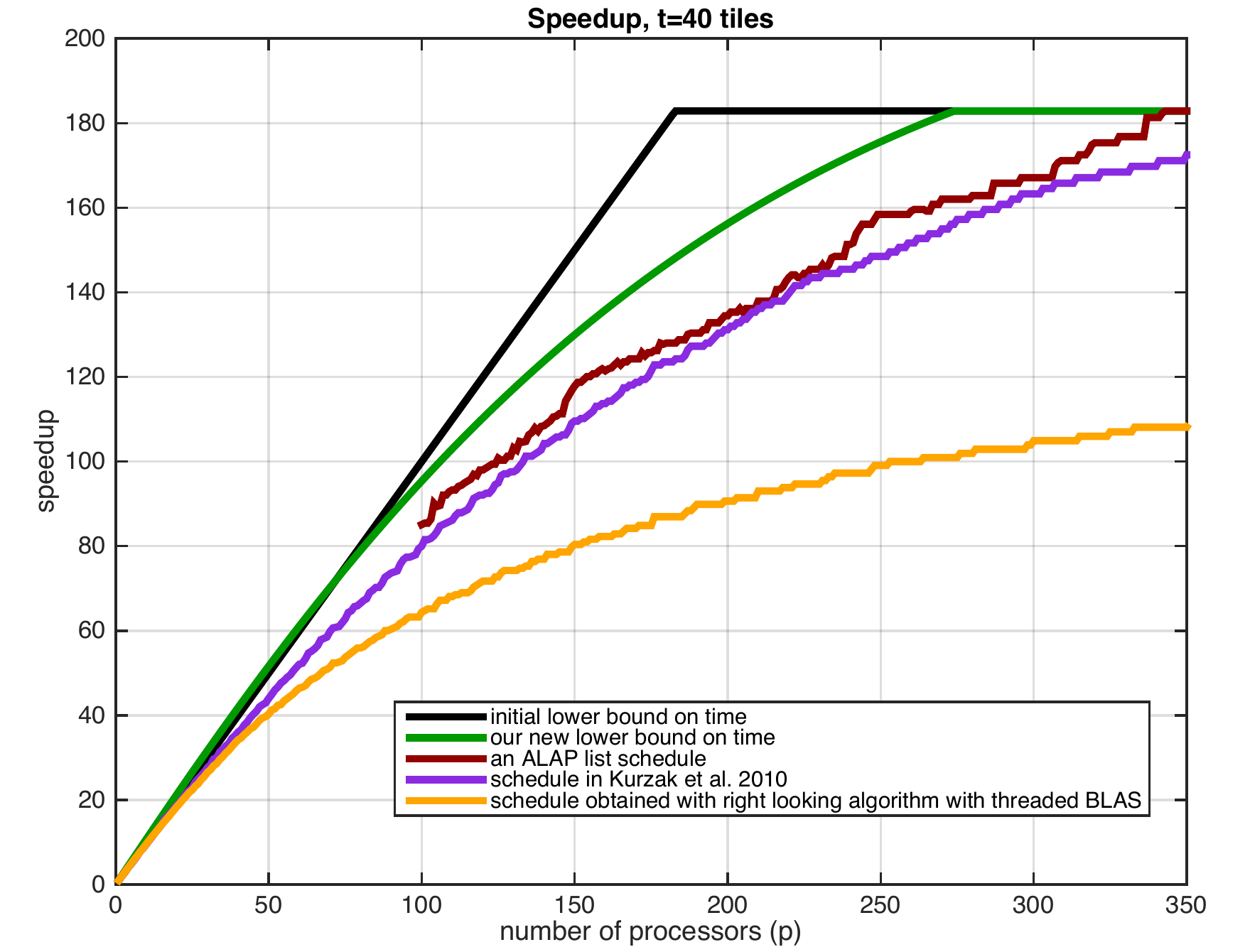}
\caption{
Speedup for a tiled Cholesky factorization with t=40 tiles.  We plot two upper
bounds (black and green curves) on the set of speedups obtained for all
schedules. And we plot three speedups (red, purple and yellow curves) obtained
by three schedules.
}

\label{fig1}
\end{figure}

\section{Conclusion}

We have analyzed the tiled Cholesky factorization and improved existing lower
bounds on the execution time of any schedule, with a technique that benefits
from the structure of dependence graph of the tiled Cholesky factorization. 

We took advantage of our observation that the tiled Cholesky factorization is
much better behaved when we schedule it with an ALAP (As Late As Possible)
heuristic than an ASAP (As Soon As Possible) heuristic. 

We believe that our theoretical results will help practical scheduling studies
of the tiled Cholesky factorization. Indeed, our results enable to better
characterize the quality of a practical schedule with respect to an optimal
schedule.

We also believe that our technique is generalizable to many tile algorithms, in
particular LU and QR. It is clear that many linear algebra operations would
benefit from the ALAP scheduling strategy.  Also, we can easily change the
weight of the tasks in our study to better represent the time of the kernels
(as opposed to the number of flops of the kernels) on a given architecture.

There are two questions left open by our work. First we do not have satisfying
closed form formula for schedules on $p$ processors. Indeed, in
Section~\ref{sec:tightness}, we relied on simulation to plot the speedup of an
ALAP schedule and of the schedule from Kurzak et
al.~\cite{journals/concurrency/KurzakLDB10}. We do not have closed form formula
for these. Also, while we made significant progress, there is still a gap between our lower and upper bounds
and it seems worthwhile to further our work and close this gap
in an asymptotic sense.


\begin{thebibliography}{10}
\providecommand{\url}[1]{#1}
\csname url@samestyle\endcsname
\providecommand{\newblock}{\relax}
\providecommand{\bibinfo}[2]{#2}
\providecommand{\BIBentrySTDinterwordspacing}{\spaceskip=0pt\relax}
\providecommand{\BIBentryALTinterwordstretchfactor}{4}
\providecommand{\BIBentryALTinterwordspacing}{\spaceskip=\fontdimen2\font plus
\BIBentryALTinterwordstretchfactor\fontdimen3\font minus
  \fontdimen4\font\relax}
\providecommand{\BIBforeignlanguage}[2]{{%
\expandafter\ifx\csname l@#1\endcsname\relax
\typeout{** WARNING: IEEEtran.bst: No hyphenation pattern has been}%
\typeout{** loaded for the language `#1'. Using the pattern for}%
\typeout{** the default language instead.}%
\else
\language=\csname l@#1\endcsname
\fi
#2}}
\providecommand{\BIBdecl}{\relax}
\BIBdecl

\bibitem{Gustavson2009}
F.~Gustavson, L.~Karlsson, and B.~K{\aa}gstr\"{o}m, ``Distributed {SBP}
  {Cholesky} factorization algorithms with near-optimal scheduling,'' \emph{ACM
  T. Math. Software}, vol.~36, no.~2, pp. 1--25, 2009.

\bibitem{Kurzak2008}
J.~Kurzak, A.~Buttari, and J.~Dongarra, ``Solving systems of linear equations
  on the {CELL} processor using {Cholesky} factorization,'' \emph{IEEE Trans.
  Parallel Distrib. Syst.}, vol.~19, no.~9, pp. 1175--1186, 2008.

\bibitem{Chan:2008}
E.~Chan, F.~G. Van~Zee, P.~Bientinesi, E.~S. Quintana-Ort{\'\i},
  G.~Quintana-Ort{\'\i}, and R.~van~de Geijn, ``Supermatrix: a multithreaded
  runtime scheduling system for algorithms-by-blocks,'' in \emph{PPoPP '08:
  Proceedings of the 13th ACM SIGPLAN Symposium on Principles and practice of
  parallel programming}.\hskip 1em plus 0.5em minus 0.4em\relax New York, NY,
  USA: ACM, 2008, pp. 123--132.

\bibitem{Agullo2010Comparison}
E.~Agullo, B.~Hadri, H.~Ltaief, and J.~Dongarra, ``Comparative study of
  one-sided factorizations with multiple software packages on multi-core
  hardware,'' in \emph{SC'09. ACM/IEEE Conference on Supercomputing, Portland,
  OR, November}, 2009.

\bibitem{Song2009}
F.~Song, A.~YarKhan, and J.~Dongarra, ``Dynamic task scheduling for linear
  algebra algorithms on distributed-memory multicore systems,'' in \emph{SC'09.
  ACM/IEEE Conference on Supercomputing, Portland, OR, November}, 2009.

\bibitem{agullo2015bridging}
E.~Agullo, O.~Beaumont, L.~Eyraud-Dubois, J.~Herrmann, S.~Kumar, L.~Marchal,
  and S.~Thibault, ``Bridging the gap between performance and bounds of
  {C}holesky factorization on heterogeneous platforms,'' in \emph{Heterogeneity
  in Computing Workshop 2015}, 2015.

\bibitem{journals/concurrency/KurzakLDB10}
J.~Kurzak, H.~Ltaief, J.~Dongarra, and R.~M. Badia, ``Scheduling dense linear
  algebra operations on multicore processors.'' \emph{Concurrency and
  Computation: Practice and Experience}, vol.~22, no.~1, pp. 15--44, 2010.

\bibitem{badia2009}
R.~M. Badia, J.~R. Herrero, J.~Labarta, J.~M. P{\'e}rez, E.~S.
  Quintana-Ort{\'{\i}}, and G.~Quintana-Ort{\'{\i}}, ``Parallelizing dense and
  banded linear algebra libraries using {SMPSs},'' \emph{Concurrency and
  Computation: Practice and Experience}, vol.~21, no.~18, pp. 2438--2456, 2009.

\bibitem{dague:2012}
G.~Bosilca, A.~Bouteiller, A.~Danalis, T.~Herault, P.~Lemarinier, and
  J.~Dongarra, ``{DAGuE}: A generic distributed dag engine for high performance
  computing,'' \emph{Parallel Computing}, vol.~38, no. (1-2), pp. 37--51, 2012.

\bibitem{kaapi:2007}
T.~Gautier, X.~Besseron, and L.~Pigeon, ``{KAAPI}: A thread scheduling runtime
  system for data flow computations on cluster of multi-processors,'' in
  \emph{PASCO'07}, London, Ontario, Canada, Jul. 2007.

\bibitem{quark:2011}
A.~YarKhan, J.~Kurzak, and J.~Dongarra, ``{QUARK} users' guide: Queueing and
  runtime for kernels,'' University of Tennessee, Innovative Computing
  Laboratory, Tech. Rep. ICL-UT-11-02, 2011.

\bibitem{starpu:2011}
C.~Augonnet, S.~Thibault, R.~Namyst, and P.-A. Wacrenier, ``{StarPU:} a unified
  platform for task scheduling on heterogeneous multicore architectures,''
  \emph{Concurrency and Computation: Practice and Experience, Special Issue:
  Euro-Par 2009}, vol.~23, pp. 187--198, Feb. 2011.

\bibitem{smpss:2007}
J.~M. P\'erez, R.~M. Badia, and J.~Labarta, ``A flexible and portable
  programming model for {SMP} and multi-cores,'' Barcelona Supercomputing
  Center – Centro Nacional de Supercomputacio{\'n}, Tech. Rep., Jun. 2007.

\bibitem{qqgzc:2009}
E.~S. Quintana-Ort{\'\i}, G.~Quintana-Ort{\'\i}, R.~A. van~de Geijn, F.~G.
  Van~Zee, and E.~Chan, ``Programming matrix algorithms-by-blocks for
  thread-level parallelism,'' vol.~36, no.~3, 2009.

\bibitem{COSNARD1988275}
M.~Cosnard, M.~Marrakchi, Y.~Robert, and D.~Trystram, ``Parallel gaussian
  elimination on an {MIMD} computer,'' \emph{Parallel Computing}, vol.~6,
  no.~3, pp. 275--296, 1988.

\bibitem{Robert1989159}
Y.~Robert and D.~Trystram, ``Optimal scheduling algorithms for parallel
  gaussian elimination,'' \emph{Theoretical Computer Science}, vol.~64, no.~2,
  pp. 159 -- 173, 1989.

\bibitem{j38}
M.~Marrakchi and Y.~Robert, ``Optimal algorithms for gaussian elimination on a
  {MIMD} computer,'' \emph{Parallel Computing}, vol.~12, pp. 183--194, 1989.

\end{thebibliography}
\end{document}